\documentclass[a4paper]{article}

\usepackage{graphicx}
\usepackage{array,colortbl}
\usepackage{amsthm,amsmath,amsfonts,amssymb} 
\usepackage{subcaption,tabularx,booktabs}
\usepackage{bbm}
\usepackage{url}
\usepackage{algorithmic}
\usepackage[margin=3.5cm]{geometry}
\usepackage{natbib}

\newtheorem{theorem}{Theorem}[section]
\newtheorem{lemma}[theorem]{Lemma}

\newtheorem{definition}{Definition}[section]

\newcommand{\Z}{{\mathbb{Z}}}

\newcommand{\PP}{{\mathbb{P}}}

\def\Pois{{\operatorname{Pois}}}

     \DeclareMathOperator{\cov}{Cov}

\providecommand{\extr}{\operatorname*{extr}}

\DeclareSymbolFont{bbold}{U}{bbold}{m}{n}
\DeclareSymbolFontAlphabet{\mathbbold}{bbold}
 
\usepackage[colorlinks=true,linkcolor=black,citecolor=black,urlcolor=black]{hyperref}
\makeatletter   \providecommand*{\toclevel@author}{999}
  \providecommand*{\toclevel@title}{0}
\makeatother

\usepackage{graphicx,epstopdf,float,caption,subcaption,wrapfig}
\DeclareGraphicsExtensions{.eps}

\usepackage{hyperref}
\hypersetup{
	colorlinks=true
}

\usepackage{algorithm,algorithmic}

\usepackage{tikz}
\usetikzlibrary{graphs}
\usetikzlibrary{calc}

\newcommand{\numberextremepoints}{N}

\newcommand{\dimensionofcorrelationspace}{M}

\newcommand{\dimensionofmpp}{J}

\newcommand{\corr}{\text{corr}}

\newcommand{\MMhigherdimenprobdistr}{p\,(i_1,\dots,i_\dimensionofmpp)}

\def\eqod{{\stackrel{{\cal D}}{=}}}

\usepackage{authblk}

\begin{document}
	\title{Correlated Multivariate Poisson Processes and Extreme Measures}
	\author[1]{Michael Chiu}
	\author[2]{Kenneth R. Jackson\thanks{This research was supported in part by the Natural Sciences and Engineering Research Council (NSERC) of Canada}}
	\affil{Department of Computer Science, University of Toronto}
	\author[3]{Alexander Kreinin\thanks{Corresponding author.}}
	\affil{Quantitative Research, Risk Analytics, IBM Canada}
	\affil[ ]{\texttt{\{chiu,krj\}@cs.toronto.edu, alex.kreinin@ca.ibm.com}}
	\date{August 29, 2017}

	\maketitle

	\abstract{Multivariate Poisson processes have many important applications in Insurance, Finance, and many other areas of Applied Probability. In this paper we study the backward simulation approach to modelling multivariate Poisson processes and analyze the connection to the extreme measures describing the joint distribution of the processes at the terminal simulation time.}

	\section{Introduction}
\label{sec:intro}
Analysis and simulation of dependent Poisson processes is an important problem having many applications in Insurance, Finance, Operational Risk modelling and many other areas (see \cite{Aue}, \cite{Bock2}, \cite{Chav},  \cite{kreinin2}, \cite{EmbPuc}, \cite{Panj}, \cite{Shev} and references therein). In the modelling of multivariate Poisson processes, the specification of the dependence structure is an intriguing problem. In some applications, such as Operational Risk, the realized correlations between components of multivariate Poisson Processes exhibit negative correlations that cannot be ignored, as exemplified in the correlation matrix below.
\begin{equation*}
\begin{bmatrix}
1.0 & 0.14 & 0.29 & 0.32 & 0.15 & 0.16 & 0.03  \\
0.14 & 1.0 & 0.55 & -0.12 & 0.49 & 0.52  & -0.16   \\
0.29 & 0.55 & 1.0 & 0.11 & 0.27 & 0.17 & -0.31  \\
0.32 & -0.12 & 0.11 & 1.0 & -0.12 & -0.23 & 0.19   \\
0.15 & 0.49 & 0.27 & -0.12 & 1.0 & 0.49 & -0.17   \\
0.16 & 0.52 & 0.17 & -0.23 & 0.49 & 1.0 & -0.02   \\
0.03 & -0.16 & -0.31 & 0.19 & -0.17 & -0.02 & 1.0   \\
\end{bmatrix}
\end{equation*}

In the literature, several different bivariate processes with Poisson marginal distributions are available for applications in actuarial science and quantitative risk management. One of the most popular models is the common shock model \cite{CSM} where several common Poisson processes drive the dependence between the components of the multivariate Poisson process. The resulting correlation structure is time invariant and cannot exhibit negative correlations in this case.  

An alternative, more flexible approach to this problem is based on the Backward Simulation 
(BS) introduced in \cite{kreinin} for the bivariate Poisson processes. The BS of correlated Poisson processes 
and  an approach to the calibration problem using transformations of Gaussian variables was proposed in
\cite{kreinin2}. In \cite{kreinin}, 
the idea of BS was extended to the class of multivariate processes containing both Poisson and Wiener components. 
It was also proved that the linear time structure of correlations is observed both in the Poisson and 
the Poisson-Wiener model. Further steps in the bivariate case were proposed in \cite{TBak} where the 
BS was combined with copula functions. This method allows one to extend the correlation pattern by using 
the Marshall-Olkin type copula functions that are simple to simulate.

In this paper, we continue the analysis and development of the BS method for the class of multivariate Poisson processes. By the multivariate Poisson process, we understand any vector-valued process such that all its components are (single-dimensional) Poisson processes. The idea of our approach is to use the relationship between the extreme measures describing the joint distribution with maximal or 
minimal correlation coefficient of the components of the multivariate process at the terminal simulation time and the time structure of correlations. We describe the class of admissible correlation structures given parameters of the marginal Poisson processes and exploit convex combinations of the extreme measures to represent the multivariate Poisson process with given correlations of the components. We believe that our approach can simplify the solution to the calibration problem and extend the variety of the correlation patterns of the multivariate Poisson processes. 

There is a connection between our problem and the Optimal Transport literature (see \cite{MKP1} for a general overview of the area and \cite{MKP_ruschendorf_1,MKP_ruschendorf_2} for a more probabilistic focus). Our computation of the extreme measures at the terminal simulation time can be viewed as a solution to a special multi-objective Monge-Kantorovich Mass Transportation Problem (MKP), with quadratic cost functions. However, this connection is not discussed in the present paper. In this paper, we are mainly concerned with the construction of the multivariate Poisson processes.

The rest of the paper is organized as follows. In Section \ref{chiu_sec:em_monot_2d} we begin 
by discussing the background and motivation for the 2-dimensional problem. We introduce extreme measures and generalize the 
results of the bivariate problem to higher dimensions in Section \ref{sec:ejd_hd}.  In Section~\ref{sec:algor} we describe a general algorithm for the computation of the joint distribution of the extreme measures.
Section \ref{sec:calibration} is concerned with the calibration problem. We   discuss 
the simulation problem in Section \ref{sec:simulation} and propose a Forward-Backward extension 
of the BS method. The paper is concluded with some directions for future research in Section \ref{sec:chiu_future_work}.
 	\section{Extreme Measures and Monotonicity of the Joint Distributions }
\label{chiu_sec:em_monot_2d} 
We begin with a description of the Common Shock Model (CSM) \cite{CSM} and the motivation of the approach proposed in \cite{kreinin2}. Afterwards, we discuss the results obtained in \cite{kreinin} for the case of two Poisson processes and describe the computation of the extreme measures in the case $J=2$. 

The CSM has become very popular within actuarial applications as well as in Operational Risk modeling \cite{Dian}. This model is based on the following idea. Suppose we want to construct two dependent Poisson processes. Consider three independent Poisson processes $\nu^{(1)}_t$, $\nu^{(2)}_t$, $\nu^{(3)}_t$ with the  intensities $\lambda_1$, $\lambda_2$, $\lambda_3$. Let $N^{(1)}_t = \nu^{(1)}_t + \nu^{(2)}_t$ and  $N^{(2)}_t = \nu^{(3)}_t + \nu^{(2)}_t$, which are also Poisson processes, formed by the superposition operation. Then, the Poisson processes $N^{(1)}_t$ and $N^{(2)}_t$ are dependent with the Pearson correlation coefficient
\begin{equation*}
    \rho(N^{(1)}_t,N^{(2)}_t) = \frac{\lambda_2}{\sqrt{(\lambda_1 + \lambda_2)(\lambda_2 + \lambda_3)}}.
\end{equation*}
Clearly, the correlation coefficient can only be positive.

A more advanced approach to the construction of negatively correlated Poisson processes is based on the idea of the backward simulation of the Poisson processes  \cite{kreinin}. 
The conditional distribution of the arrival moments of a Poisson process, conditional on the value of the process at the terminal simulation time, $T$, is uniform. Then, using a joint distribution maximizing or minimizing correlation between the components at time, $T$, one can construct a Poisson process with a linear time structure of correlations in the interval $t\in [0, T]$. Thus, the problem of constructing the $2$-dimensional Poisson process with the extreme correlation of the components at time $T$ is reduced to that of random variables having Poisson  distributions with the parameters $\lambda T$ and $\mu T$, where $\lambda$ and $\mu$ are parameters of the processes. It is not difficult to see that maximization (minimization) of the correlation coefficient of two random variables (r.v.), $X$ and $Y$, given their marginal distributions, is equivalent to maximization (minimization) of $\mathbb{E}[XY]$, if the r.v. have finite first and second moments and positive variances.

The admissible range of the correlation coefficients can be computed using the Extreme Joint Distributions (EJD) Theorem in \cite{kreinin} (see Theorem~\ref{thm:ejd_2d} in this section). The key statement, the characterization of the EJDs, is equivalent to the Frechet-Hoeffding theorem \cite{Fre} for distributions on the positive quadrant of the two-dimensional lattice, $\Z^{(2)}_{+}=\{(i, j): i,j=0,1,2,\dots\}$. However, taking into account the numerical aspect of the problem, we prefer to use equations, derived in \cite{kreinin}, written in terms of the probability density function, not in terms of the cumulative distribution function. Given marginal distributions of the non-negative, integer-valued random variables $X_1$ and $X_2$, with finite first and second moments, there exist two joint distributions, $F^{*}(i, j)$ and $F^{**}(i, j)$ minimizing and maximizing the correlation, $\rho=\corr(X_1, X_2)$, respectively.
\begin{definition}  \label{def:chiu_extrem_measure_2d} The probability measures 
corresponding to the joint distributions $F^{*}$ and $F^{**}$  are called extreme probability measures.
\end{definition}
The EJD theorem in \cite{kreinin} allows one to construct the extreme measures $p^*$ and $p^{**}$, given marginal distributions of $X_1$ and
$X_2$, with the minimal negative correlation $\rho^*$ and maximal positive correlation $\rho^{**}$, respectively.  The extreme correlation coefficient uniquely defines the extreme measure.
   
Given a probability measure, $p$, corresponding to the joint 
distribution of the vector $(X_1, X_2)$ on  $\Z^{(2)}_{+}$ we define a functional $f_\rho(p)=\corr(X_1, X_2)$. Then we have $\rho^* =f_\rho( p^*) ,  \text{and}\,\,    \rho^{**} = f_\rho( p^{**})$. This functional $f_{\rho}$ preserves the convex combination property. Indeed, taking a convex combination of the extreme measures, 
$p= \theta p^* + (1-\theta)p^{**}$, $(0\le \theta\le 1)$, we obtain   
\begin{equation}
    f_\rho (p) = \theta f_\rho(p^*) + (1-\theta) f_\rho(p^{**}).  \label{eq:lin_rho}
\end{equation}

Thus, for any $\rho \in [\rho^*,\rho^{**}]$, we can find a probability measure $p$ for a joint distribution of the vector
 $(X_1,X_2)$ such that $f_\rho(p) = \corr(X_1,X_2) = \rho$ and $p$ has the required marginal distributions for $X_1$ and $X_2$.

\section*{Connection to Optimization Problem}
Computation of the  extreme measures in  the case $J=2$ was accomplished in \cite{kreinin} using a very efficient EJD algorithm having linear complexity with respect to the number of points in the support of the marginal distributions. It is interesting to note that this algorithm is applicable to a more general class of linear optimization problems on a lattice. In the case $J>2$, the corresponding optimization problem becomes multi-objective with $M=J(J-1)/2$ objective functions. Let us first recall the case $J=2$.

Let $(X_1, X_2)$ be a random vector with support $\mathbb{Z}^{(2)}_{+}$ and given marginal probabilities $\PP(X_1=i)=P_1(i)$ and $\PP(X_2=j)=P_2(j)$. Denote
\begin{equation*}
    h(p) := \mathbb{E}[X_1 X_2] = \sum^\infty_{i=0}\sum^\infty_{j=0} ij\,p(i,j)
    \label{eq:objective_fn_2D} 
\end{equation*} 
where $p(i,j)=\PP( X_1 = i, X_2 = j)$.
The measure $p^{**}$ is the solution to the problem $ h(p)\to \max$ with the constraints shown below in (\ref{eq:optimization_problem_2D}) on the marginal distributions of $p^{**}$. Similarly, the extreme measure $p^{*}$ is the solution to the optimization problem $ h(p)\to \min$ with the same constraints \cite{kreinin}. For the sake of brevity, we write these two problems as
\begin{align}
    & \quad h(p) \rightarrow \extr  \label{eq:optimization_problem_2D} \\
    \text{subject to} \nonumber \\   
    & \quad \sum^\infty_{j=0} p(i,j) = P_1(i), \quad i=0,1,\dots \nonumber \\
    & \quad \sum^\infty_{i=0} p(i,j) = P_2(j), \quad j=0,1,\dots \nonumber \\ 
   & \quad p(i,j) \geq 0 \quad i,j=0, 1, 2, \dots \nonumber
\end{align}
where  $\sum_{i=0}^\infty\limits P_1(i)=\sum_{j=0}^\infty\limits P_2(j)=1$.
The symbol $\extr$ denotes $\max$ in the case of measure $p^{**}$ and $\min$ in the case of $p^*$. It is not difficult to see that Problem~(\ref{eq:optimization_problem_2D}) is infinite dimensional; its numerical solution requires construction of the compact subset of the lattice for the computation of the approximate solution  \cite{kreinin}.

A solution to the infinite dimensional optimization problem (\ref{eq:optimization_problem_2D}) is the joint distribution describing one of the extreme measures, given the marginal distributions of the random variables. The EJD algorithm discussed in \cite{kreinin} allows one to find a unique solution to the problem to any user specified accuracy. Taking the marginal distributions to be Poissonian, we find the extreme measures, $p^*$ and $p^{**}$, describing the joint distribution of the processes, $N_T=(N^{(1)}_T,  N^{(2)}_T)$ with the extreme correlation of the components at time $T$. 

The convex combination of these measures can be calibrated to the desired value of the correlation coefficient, $\rho$. Then, applying the BS method we obtain the sample paths of the processes. Note that the EJD algorithm is applicable to a more general class of linear optimization problems: there is no need to assume normalization conditions as long as  $P_1(i)\ge 0$ and $P_2(j)\ge 0$ for all $i \geq 0$ and $j \geq 0$ and these functions are integrable: $\sum_{i=0}^\infty\limits P_1(i)< \infty$ and $\sum_{j=0}^\infty\limits P_2(j)< \infty$.

\section*{Monotone Distributions}
\label{sec:montone_distr}
Extreme measures are closely connected to the monotone distributions in the case $J=2$. It was proved in \cite{kreinin} that
the joint distribution is comonotone in the case of maximal correlation and antimonotone in the case of minimal (negative) correlation. Let us  review  the 
properties of extreme measures used in what follows. 

Consider a set $\mathcal{S} = \{ s_n \}_{n \geq 0}$, where $s_n = (x_n,y_n) \in \mathbb{R}^2$. Define the two subsets $\mathcal{R}_+ = \{(x,y) \in \mathbb{R}^2 : x \cdot y \geq 0 \} \qquad \text{and} \qquad \mathcal{R}_- = \{(x,y) \in \mathbb{R}^2 : x \cdot y \leq 0 \}$.
\begin{definition}
    A set $\mathcal{S} = \{s_n\}_{n\geq 0} \subset \mathbb{R}^2$ is comonotone if $\forall \, i, j$,  $s_i - s_j \in \mathcal{R}_+ $. Similarly, $\mathcal{S}$ is antimonotone if $\forall \, i, j$,  $s_i - s_j \in \mathcal{R}_-$.
    \label{def:monotone_sets_2d}
\end{definition}
\begin{definition}[Monotone distributions]
We say that a distribution $P$ is comonotone (antimonotone) if its 
support is a comonotone (antimonotone) set.
\end{definition}
It is also useful to recall the following classical statement on monotone sequences of real numbers, usually attributed to Hardy.\footnote{
This result motivates and is used in the proof of Theorem 2.2 in \cite{kreinin} and provides an explanation as to why one coordinate of the support always increases (decreases) in the comonotone (antimonotone) case.
}

Consider two vectors $x \in \mathbb{R}^N$ and $y \in \mathbb{R}^N$. Their
inner product is
\begin{equation*}
    \langle x, y \rangle := \sum^N_{k=1} x_k y_k
\end{equation*}
Denote by  $\mathfrak{S}_N$ the set of all permutations of $N$ elements.
\begin{lemma}\label{lem:H}
    For any monotonically increasing sequence, $x_1 \leq x_2 \leq \dots \leq x_N$ and a vector $y \in \mathbb{R}^N$, there exist permutations $\pi_+$ and $\pi_-$ solving the optimization problems
\begin{equation*}
  \langle x, \pi_+ y \rangle = \max_{\pi \in \mathfrak{S}_N} \langle x, \pi y \rangle
    \end{equation*}
    and
    \begin{equation*}
\langle x, \pi_- y \rangle = \min_{\pi \in \mathfrak{S}_N} \langle x, \pi y \rangle
    \end{equation*}
    The permutations $\pi_+$ and $\pi_-$ sort vectors in ascending and descending order, respectively.
\end{lemma}
Lemma~\ref{lem:H}  motivates the introduction of monotone distributions in the $2$-dimensional case.   
\begin{theorem}[\cite{kreinin}]
    \label{thm:ejd_2d}
    The joint distribution $p^{**}$ for $X_1$ and $X_2$ having 
maximal positive correlation coefficient $\rho^{**}$, given marginal distributions $F_1(i)$ and $F_2(j)$, is comonotone. The probabilities 
$p^{**} (i,j)=\mathbb{P}(X_1=i, X_2=j)$ satisfy the equation
    \begin{align}
        p^{**} (i,j) & = [\min(F_1(i),F_2(j)) - \max(F_{1}(i-1),F_2(j-1))]^+ \quad i,j = 0,1,2,\dots
        \label{eq:mr_J}
    \end{align}
    where $[\,x\,]^+ = \max(x,0)$ and $F_i(\cdot)$ denote the marginal CDFs, with $F_i(-1) = 0$.
    
    The joint distribution $p^*$ for $X_1$ and $X_2$ having minimal negative correlation coefficient $\rho^*$ is antimonotone. In this case
    \begin{equation}
        p^*(i,j) = [\min(F_1(i),\bar{F}_2(j-1)) - \max(F_1(i-1),\bar{F}_2(j))]^+ \quad i,j = 0,1,2,\dots 
        \label{eq:max_negat}
    \end{equation} 
    where $\bar{F}_i(j) = 1 - F_i(j)$ and $\bar{F}_i(-1)=1$.
\end{theorem} 
Theorem~\ref{thm:ejd_2d} is equivalent to the Frechet theorem in the case the  marginal distributions  are discrete.
 
The case of the Poisson marginal distributions is a particular case of Theorem~\ref{thm:ejd_2d}. This result is applicable to much more general classes of distributions. In particular, one can describe the joint probabilities corresponding to $p^*$ and $p^{**}$ in the case the components of the vector have a negative binomial distribution. The EJD algorithm for computation of the joint probabilities is also applicable to more general cases. If both marginal distributions have finite second moments, the joint distribution can be approximated to any user specified accuracy.
 
\section{Extreme Measures in Higher Dimensions}\label{sec:ejd_hd}
Let us now generalize the main result, Theorem~\ref{thm:ejd_2d},
discussed in Section~\ref{chiu_sec:em_monot_2d}.
We consider a random vector $\vec{X}=(X_1,\dots,X_J)$ on a positive quadrant of the $J$-dimensional  lattice, $\Z^{(J)}_+$. Each coordinate of $\vec{X}$ has a discrete distribution with the support $\mathbb{Z}_+$. We also assume that each random variable $X_k$, $k=1,2,\dots,J$, has finite second moment and its variance is positive. In this case, the correlation coefficients, $\rho_{k,l}=\corr(X_k, X_l)$, are defined for all $1\le k\le l\le J$.
We denote the marginal distribution of the r.v. $X_k$ by $F_k$:
$$ F_k(i)=\mathbb{P}(X_k\le i), \quad i\in \mathbb{Z}_+; \quad k=1, 2, \dots, J. $$ 
Let us now define the extreme measures on the $J$-dimensional lattice. If $J=2$, the extreme measures are described by the joint distribution maximizing and minimizing the correlation coefficient of $X_1$ and $X_2$; the corresponding probability density functions satisfy Theorem~\ref{thm:ejd_2d}. If the number of components $J\ge 3$, the definition of the extreme measure is less obvious. 

Denote the (joint) distribution function of $\vec{X}$ by $F(\vec{i})$:
$F(i_1, i_2, \dots, i_J)=\PP(X_1\le i_1, X_2\le i_2, \dots, X_J\le i_J)$ and the corresponding probability density function by $p(\vec{i})$. By $p_{k,l}(i_k, i_l)$ we denote the  probability density function of the $2$-dimensional projection, $(X_k, X_l)$ of  $\vec{X}$, $(1\le k < l\le J)$:
\begin{equation*}
	p_{k,l}(i_k,i_l) = P(X_k = i_k, X_l = i_l)
\end{equation*}
\begin{definition}
    We say that the density $p(\vec{i})$,
$$ p(i_1, \dots, i_J)=\mathbb{P}(X_1= i_1, \dots, X_J= i_J), \quad i_k\in  
\mathbb{Z}_+, \quad k=1, 2, \dots, J $$
determines an extreme measure on the $J$-dimensional lattice if and only if
for all  $k$ and $l$, $(1\le k\le l\le J)$, the associated density $p_{k,l}$ 
determines an extreme measure on $\Z^{(2)}_{+}$ in the sense of Definition~\ref{def:chiu_extrem_measure_2d}.
    \label{def:extr_meas_gen}
\end{definition}   
Our goal is to describe the extreme measures given the marginal distributions, $F_k$, and compute the associated extreme correlation matrices,
$\mathbf{\rho}=[\rho_{k,l}]$.  Let us first find the number of  extreme  measures.
\begin{lemma}\label{lem:numb_ed}
For any given set of marginal distributions, $F_k$,  on $\mathbb{Z_+}$
$(k=1, 2, \dots, J)$ the number of extreme measures is $N=2^{J-1}$.
\end{lemma}
\begin{proof}
The proof of Lemma~\ref{lem:numb_ed} for $J=2$ is obvious. Let us prove  it
for $J\ge 3$.
For each $2$-dimensional  projection $(X_k, X_l)$, the corresponding joint distribution 
should be either comonotone or antimonotone.  Take the first r.v, 
$X_1$, and form the first group of random variables from the set
$X_2$, $X_3$, $\dots, X_J$, that are comonotone with $X_1$. Denote the number
of comonotone r.v.,\ by $J_c$. The number of r.v.\ antimonotone with $X_1$,  satisfies
$$J_a=J-1 - J_c.$$
The total number of partitions of the number $J-1$ in the additive form,
$ J-1= J_a + J_c$,
is $N=2^{J-1}$.  Clearly,  $N$ does not depend 
on the choice of  the  first r.v.
\end{proof}
Let us now introduce the  monotonicity  structure of the extreme  measures.
Take  the first r.v., $X_1$ and consider the r.v. $X_2$, $X_3$, $\dots X_J$.
Define the vector of binary variables $\vec{e^n}=(e_1, e_2, \dots, e_J)$ 
such that $e_1=0$, and for $j=2, 3, \dots, J$, $n=1,\dots,N$.
$$ e_j=
\begin{cases} 1,&\text{if $X_1$ and $X_j$ are antimonotone,}\cr
                      0,&\text{if $X_1$ and $X_j$ are comonotone.}\cr
\end{cases}
$$
We call $\vec{e^n}$ the monotonicity vector corresponding to the $n$-th extreme measure; its components are called monotonicity indicators. Figure \ref{chiu_fig:extremal_structure} illustrates this concept. In this example, all coordinates but the last are comonotone with the first r.v., $X_1$. The last coordinate, $X_J$ is antimonotone. The monotonicity indicators in this case are $e_k=0$, for $k=1, 2, \dots, J-1$, and $e_J=1$.

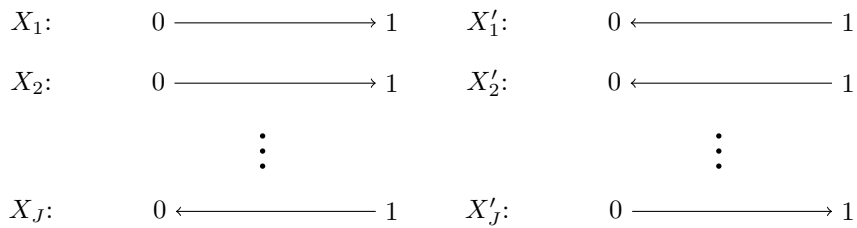
\begin{figure}[h]
\hspace{1cm}
\begin{tikzpicture}
\node (a) at (0,0) {$X_1$: \hspace{1cm} 0};
\node (b) at (4,0) {1};

\node (c) at (0,-0.8) {$X_2$: \hspace{1cm} 0};
\node (d) at (4,-0.8) {1};

\node (e) at (0,-2.5) {$X_J$: \hspace{1cm} 0};
\node (f) at (4,-2.5) {1};

\node (a2) at (6,0) {$X'_1$: \hspace{1cm} 0};
\node (b2) at (10,0) {1};

\node (c2) at (6,-0.8) {$X'_2$: \hspace{1cm} 0};
\node (d2) at (10,-0.8) {1};

\node (e2) at (6,-2.5) {$X'_J$: \hspace{1cm} 0};
\node (f2) at (10,-2.5) {1};		

\fill (2.3,-1.5) circle (1pt);
\fill (2.3,-1.7) circle (1pt);
\fill (2.3,-1.9) circle (1pt);

\fill (8.3,-1.5) circle (1pt);
\fill (8.3,-1.7) circle (1pt);
\fill (8.3,-1.9) circle (1pt);

\graph{(a)->(b)};
\graph{(c)->(d)};
\graph{(e)<-(f)};

\graph{(a2)<-(b2)};
\graph{(c2)<-(d2)};
\graph{(e2)->(f2)};

\end{tikzpicture}
\captionsetup{width=0.85\textwidth}
\caption{Monotonicity structure of an extreme measure. Each distribution is represented by an arrow having unit length. All arrows associated with $X_2$,   $\dots, X_{J-1}$ are oriented in the same direction as the arrow representing $X_1$. The last arrow pointing in the opposite direction indicates antimonotonicity of the random variables $X_1$ and $X_J$. The monotonicity structure on the right has all of its arrows reversed compared to that on the left. However, note that they both represent the same monotonicity structure.}
\label{chiu_fig:extremal_structure}
\end{figure} \subsection*{Optimization Problem: $J\ge 3$.}

Since each  $2$-dimensional  projection of the random vector $\vec{X}$ is associated with an extreme measure, the optimization problem in this case is multiobjective. The number of optimization criteria  is  $M=J(J-1)/2$, one for each pair of r.v.s\ $(X_i,X_j), \,\,\, 1 \leq i < j \leq J.$ The number of constraints is  equal to the number of marginal distributions, $J$. The variables in this problem are the probabilities $$ p(\vec{i})=\PP(X_1=i_1, X_2=i_2, \dots, X_J=i_J), \quad  i_j\in \mathbb{Z_+}, $$ and, therefore, must satisfy the inequalities $0\le p(\vec{i})\le 1$. 

Let us  define the set of integers 
$$\mathcal{I}_{k}=\{j: 1\le j\le J, \quad j\neq k, \}$$
and
$$\mathcal{I}_{k,l}=\{j: 1\le j\le J, \quad j\neq k, \quad j\neq l\}.$$
Then  the marginal probabilities, $P_k(i_k)$, can be written as
$$ p_k(i_k) = \mathbb{P}(X_k = i_k) = \sum_{j\in \mathcal{I}_k} \sum_{i_j=0}^\infty p(i_1, \dots, i_J), 
\quad i_k\in \Z_+. $$
The probabilities of the $2$-dimensional projections
$$ p_{k, l}(i_k, i_l)=\PP(X_k=i_k, X_l=i_l), \quad k,l= 1,2,\dots, J, ~
k\neq l, ~ k,l\in \Z_+,$$
are computed as 
$$ p_{k, l}(i_k, i_l)= \sum_{j\in \mathcal{I}_{k,l}} \sum_{i_j=0}^\infty 
p(i_1, \dots, i_J). $$
Similarly, the objective functions, $h_{k,l}(p)=\mathbb{E}[X_k X_l]$, take the form$$ 
h_{k,l}(p)= \sum_{i_k=1}^\infty \sum_{i_l=1}^\infty i_k i_l p_{k,l}(i_k, i_l), 
\quad 1\le k<l\le J. 
$$
The optimization problem can then be written as
\begin{align}
    &   h_{k,l}(p)  \rightarrow  \,\, \textrm{extr} \quad  1\leq k<l \leq \dimensionofmpp,  \label{eq:multi_d_prob}\\
    \text{subject to} \nonumber \\
    & \sum_{j\in \mathcal{I}_k} \sum_{i_j=0}^\infty p(i_1, \dots, i_J), = P_k(i_k)  
    \quad i_k \in \Z_+, \quad k=1,\dots,J  \nonumber\\
    &  \MMhigherdimenprobdistr \geq 0 \nonumber
\end{align}
where $P_j(\cdot)$ are given marginal probabilities $(j=1, 2, \dots, J)$.
\subsection*{The main theorem}
Let us now formulate the main result of the paper.  It is convenient to introduce
the following notation.
\begin{equation}
        \tilde{F}_j(i_j, e_j) =
        \begin{cases}
            F_j(i_j) \quad & \text{if} \,\,\, e_j = 0 \\
            1 - F_j(i_j) \quad & \text{if} \,\,\, e_j = 1
        \end{cases} 
\end{equation}
where the marginal distributions, $F_j(\cdot)$, satisfy
\begin{equation*}
F_j(i_l) = \sum_{i_k=0}^{i_l} P_j(i_k) 
\end{equation*}

\begin{theorem}[Extreme Joint Distributions in Higher Dimensions]
    \label{thm:ejd_nd}
    Given marginal distributions $F_1$, $F_2$, $\dots F_J$ on $\Z_+$ and 
    a binary vector $\vec{e^n}$,  the extreme measure  with the monotonicity
    structure $\vec{e^n}$ is defined by the probabilities        
    \begin{align}
        p^{\vec{e^n}}(\vec{i}) = \big[ \min(\tilde{F}_1 & (i_1- e_{1};e_{1}),\dots,
        \tilde{F}_J(i_J-e_{\dimensionofmpp};e_{\dimensionofmpp})) \label{eq:ejd_neq_dim} \\
        & - \max(\tilde{F}_1(i_1+(e_{1}-1);e_{1}),\dots, \tilde{F}_J(i_J+        
        (e_{\dimensionofmpp}-1);e_{\dimensionofmpp})) \big]^+ \nonumber
    \end{align}
\end{theorem}
\begin{proof} 
We give a sketch of the proof here for the general case $J \geq 2$.
A more complete proof for the case $J=2$ is given in \cite{kreinin}.

Let us first show that, if $J=2$, then  Equation~(\ref{eq:ejd_neq_dim}) is equivalent to (\ref{eq:mr_J}), in the case of maximal correlation, and to (\ref{eq:max_negat}), in the case of minimal correlation. Indeed, in the first case, the distributions of $X_1$ and $X_2$ must be comonotone. Hence, $e_1=e_2=0$ and $\tilde F_k(i, e_k)= F_k(i)$ for $k=1$ and 2 and all $i \geq 0$. In the antimonotone case, $e_1 = 0$, but $e_2=1$. Thus, $\tilde F_1(i,e_1) = F_1(i)$, but $\tilde F_2(i, e_2)= 1-F_2(i-1)$ for all $i \geq 0$.  
Therefore, Equation~(\ref{eq:ejd_neq_dim}) is equivalent to (\ref{eq:mr_J}) and (\ref{eq:max_negat}).

Let us now consider the general case, $J\ge 3$. There are two groups of the coordinates of $\vec{X}$:  comonotone and antimonotone. Denote their indices by $$ \mathcal{I_C} = \{j: e_j=0\}\,\,\, \text{and}\,\,\,  \mathcal{I_A} = \{j: e_j=1\}. $$
Let us now generate  a large sample from  the distribution  $p^{\vec{e}}$ and sort them in the ascending order with respect to the first coordinate. It was shown in \cite{kreinin} that, after sorting, the comonotone coordinates of $\vec{X}$ will be  permuted in the ascending order while the  antimonotone coordinates will be permuted in the descending order.

Suppose that the indices  $1 = k_1 < k_2< k_3<\dots<k_C$ belong to $ \mathcal{I_C}$ and the complimentary set of indices is  $\mathcal{I_A} =\{l_1, l_2, \dots, l_A\}$. A permuted sample is represented in (\ref{chiu_eq:sample_vectors_marginals}). 
\begin{align}
    X_1 &: \overbrace{0,\dots,0}^{N_{1}(0)},\dots,\overbrace{i-1,\dots,i-1}^{N_{1}(i-1)},\,\,\,\overbrace{i,i,\dots,i}^{N_{1}(i)},\dots\overbrace{k,k\dots,k}^{N_{1}(k)},\dots \nonumber \\
    X_{\,k_2} &: \underbrace{0,0,\dots,0}_{N_{k_2}(0)},\dots,
    \underbrace{i-1,\dots,i-1}_{N_{k_2}(i-1)},\,\,\underbrace{i,\dots,i}_{N_{k_2}(i)},\,\,\dots, \nonumber \\
    & \hspace{4.25cm} \vdots \label{chiu_eq:sample_vectors_marginals} \\
        X_{\,l_A} & : \dots\underbrace{k,k,\dots,k,}_{N_{\,l_A}(k)}\underbrace{k-1,\dots,k-1}_{N_{\,l_A}(k-1)},\dots\underbrace{2,2,2,\dots2}_{N_{\,l_A}(2)},\dots \nonumber
\end{align}
where $N_k(m)$ denotes the number of realizations of $m$ in the  $k$th coordinate, $X_k$, of  $\vec{X}$.   
The first position, $I_{k}^{\,C}(m)$, where the number $m$ appears in the sorted sample  of the r.v. $X_{k}$ is
$$I_{k}^{\,C}(m)= 1 + \sum_{i=0}^{m-1} N_{k}(i), \quad k \in \mathcal{I_C}. $$
The last position, $E_{k}^{\,C}(m)$, where the number $m$ appears in the sorted sample  of the   r.v. $X_{k}$ is 
$$E_{k}^{\,C}(m)=\sum_{i=0}^{m} N_{k}(i), \quad k\in \mathcal{I_C}. $$
As the sample size $N_S\to\infty$, we have
\begin{equation}
  \lim_{N_S\to\infty}  \frac{N_k(m)}{N_S} = p_k(m) \quad \textbf{a.s.} \label{eq_Nkm_as}
\end{equation}
Therefore,  for  $k\in \mathcal{I_C}$
\begin{equation}
  \lim_{N_S\to\infty}  \frac{I_k^{\,C}(m)}{N_S} = F_k(m-1) \quad \textbf{a.s.}. \label{eq_Ikm_as}
\end{equation}
and
\begin{equation}
  \lim_{N_S\to\infty}  \frac{E_k^{\,C}(m)}{N_S} = F_k(m) \quad \textbf{a.s.}. \label{eq_Ekm_as}
\end{equation}
In the case of the group of antimonotone coordinates, $l\in \mathcal{I_A}$,
the first index, $I_{l}^{\,A}(m)$, where a number $m$ appears in the sorted sample of the r.v. $X_{l}$ is 
$$I_{l}^{\,A}(m)= 1 + N_S - \sum_{i=0}^{m} N_{l}(i), \quad l \in \mathcal{I_A}. $$
The last position, $E_{l}^{\,A}(m)$, where a number $m$ appears in the sorted sample  of the  r.v. $X_{l}$ is 
$$E_{l}^{\,A}(m)= N_S - \sum_{i=0}^{m-1} N_{l}(i), \quad l\in \mathcal{I_A}. $$
As $N_S\to\infty$, we have for $l\in \mathcal{I_A}$
\begin{equation}
  \lim_{N_S\to\infty}  \frac{I_l^{\,A}(m)}{N_S} = 1- F_l(m) \quad \textbf{a.s.}. \label{eq_Ilm_as}
\end{equation}
and
\begin{equation}
  \lim_{N_S\to\infty}  \frac{E_l^{\,A}(m)}{N_S} = 1-F_l(m-1) \quad \textbf{a.s.}. \label{eq_Elm_as}
The empirical measure of the event 
\end{equation} 
$$ \{\vec{X}=\vec{i}\} =\Big\{ \bigcap_{k\in \mathcal{I_C}}  
\{X_k=i_k\} \Big\} \, \bigcap 
\Big\{ \bigcap_{l\in \mathcal{I_A}}  \{X_l=i_l\}\Big\} $$
is $\mathbf{m_{N_S}}(\{\vec{X}=\vec{i}\})$, which coincides with that of the intersection of the intervals
$$ \Big\{\bigcap_{k\in \mathcal{I_C}} [I_{k}^{\, C}(i_k), E_{k}^{\,C}(i_k)]\Big\} \bigcap  
 \Big\{\bigcap_{l\in \mathcal{I_A}} [I_{l}^{\,A}(i_l), E_{l}^{\,A}(i_l)]\Big\} $$
 The latter can be written as  follows. The right end of the intersection of the intervals is
 $$ \mathcal{R}=\min\Big( \min_{k\in \mathcal{I_C}}( E_{k}^{\,C}(i_k) ),  
               \min_{l\in \mathcal{I_A}}( E_{l}^{\,C}(i_l) ) \Big) $$
             and  the left end is
 $$ \mathcal{L} =   \max\Big(  \max_{k\in \mathcal{I_C}} (I_{k}^{\,C}(i_k)),  \max_{l\in \mathcal{I_A}} (I_{l}^{\,C}(i_l)) \Big)$$
 Then we obtain
\begin{eqnarray*} 
& &\mathbf{\mu_{N_S}}(\{\vec{X}=\vec{i}\} = \frac{(\mathcal{R}-\mathcal{L})^+}{N_S}
\label{eq_edf}       
\end{eqnarray*}
Note that the length of the intersection of intervals is $0$ in the case $\mathcal{R}\le \mathcal{L}$. 
As $N_S\to\infty$, we obtain from Equations (\ref{eq_Ikm_as})--(\ref{eq_Elm_as}) 
\begin{align*} 
\lim_{N_S\to\infty}  \mathbf{\mu_{N_S}}(\{\vec{X}=\vec{i}\} = 
     \big[ \min( & \tilde{F}_1 (i_1- e_{1};e_{1}),\dots,
        \tilde{F}_J(i_J-e_{J};e_{J}))  \\
        - & \max(\tilde{F}_1(i_1+(e_{1}-1);e_{1}),\dots, \tilde{F}_J(i_J+        
        (e_{J}-1)\big]^+.
\end{align*}
Finally, note
$$ \lim_{\, N_S\to\infty} \mathbf{\mu_{\,N_S}}(\{\vec{X}=\vec{i}\}=p^{\,\vec{e}}(\vec{X}=\vec{i}) \quad \textbf{a.s.}$$
Thus  (\ref{eq:ejd_neq_dim}) is derived and the theorem is proved.
\end{proof}

\section{EJD Algorithm in Higher Dimensions}\label{sec:algor}
\subsubsection*{Approximation of Extreme Distributions}
In practice, the marginal distributions $F_j(k), \, (j=1,\dots,J)$ must be truncated, i.e., approximated by distributions $ \tilde{F}_j(k)$ with finite support, $k\in [0, I_*]$,
such that 
\begin{equation*}
\max_{i\leq I_*} | F_j(i) - \tilde{F}_j(i) | \leq \epsilon, \quad 1 - F_j({I_*}) \leq \varepsilon, \quad \text{and for $k>I_*, \tilde{F}_j(k)=1$,}
\end{equation*}
where $F_j(n) = \sum^n_{i=0}\, p_j(i)$ and $\tilde F_j(n) = \sum^n_{i=0}\, \tilde p_j(i)$. It follows from Theorem (\ref{thm:ejd_nd}) that $\tilde{p}^{\,\vec{e^n}}(\vec{i})$ satisfies
\begin{equation}
	\sup_{i_1 \geq  0,\dots, i_J \geq 0} | \, p^{\,\vec{e^n}}(\vec{i}) - \tilde{p}^{\,\vec{e^n}}(\vec{i})\,| \leq \varepsilon. \label{ineq_p}
\end{equation}
Moreover, if the second moments of the marginal distributions are finite then for any pair of indices, $l$ and $m$, $(1\le  l\le m\le J)$, the covariance 
$\cov(X_l, X_m)$ will also be approximated
\begin{equation}
	\sup_{l, m }\Big\vert  \, \sum i_ l i_m  \Bigl( p^{\,\vec{e^n}}(\vec{i}) - \tilde{p}^{\,\vec{e^n}}(\vec{i})\Bigr) \, \Big\vert \leq 3\varepsilon.
	\label{ineq_cpv}
\end{equation}

Inequalities (\ref{ineq_p}) and (\ref{ineq_cpv}) were derived in \cite{kreinin}, in the case $J=2$, where we also explained how to choose $I_*$ given $\varepsilon$ and the second moment of the marginal distributions. The same line of arguments from \cite{kreinin} can easily be extended to the general case $J\ge 3$. These inequalities are used in the numerical example illustrating the computation of the joint probabilities of the $3$-dimensional Poisson process.

Let us now describe the Extreme Joint Distribution (EJD) algorithm, an efficient algorithm for the computation of the probabilities $p^{\vec{e^n}}(\vec{i})$ for $J\geq2$. A simpler version of this algorithm was given in \cite{kreinin} for $J=2$. The preliminary step, the truncation of the marginal distributions by distributions with finite support is identical to that in \cite{kreinin}. The main step is the recursive computation of the probabilities $p^{\,\vec{e^n}}(\vec{i})$, which can be done as described in the algorithm below. Note that, in the algorithm, $p^{\vec{e}}(\vec{x})$ is assigned a value (in Step 5) only if $\vec{x}$ is in the support of $p^{\vec{e}}$ and the support point $\vec{x}$ is saved (in Step 3).  If $\vec{x}$ is not a saved support point (i.e., not saved in Step 3), then $p^{\vec{e}}(\vec{x}) = 0$. To simplify the description of the algorithm below, we assume that all the marginal probabilities are positive.
 
\begin{tabbing}
	\hspace*{1.2cm}\= \kill
	Step 0a.	\> Set $k=0$  \\
	Step 0b.	\> 	For each $j=1:J$ \\
				\>	\quad If $e_{j} = 1$, 	\\
				\> 	\quad\quad\quad Set $F_j(i) = 1 - F_j(i)$ \\
				\> 	\quad\quad\quad Set $\Delta_j$ = -1 \quad and \quad $x_j^0 = \max \{i : P_j(i)\geq 0 \}$	\\
				\> \quad else \\
				\> \quad\quad\quad Set $\Delta_j$ = 1 \quad and \quad $x_j^0 = 0$	\\
	Step 0c. 	\> Set $z_0 = \min(F_1(0),\dots,F_\dimensionofmpp(0))$ \quad  and \quad $p^{\, \vec{e}}(x_1^0,\dots,x_J^0)=z_0$  \\
	Step 1. 	\> Set $k=k+1$ \\
	Step 2.		\> For each $j = 1:J$ \\
				\> \quad\quad If $z_{k-1} = F_j(i_j)$ for some $i_j$, \\
				\> \quad\quad\quad Set $x_j^k = i_j + \Delta_j$ \\
				\> \quad\quad else \\
				\> \quad\quad\quad Set $x_j^k = x_j^{k-1}$ \\
	Step 3. 	\> Save the $k$-th support point $\vec{x}_k = (x_1^k,\dots,x_\dimensionofmpp^k)$ \\
		Step 4. 	\> Set $z_k = \min(F_1(x_1^k),\dots,F_\dimensionofmpp(x_\dimensionofmpp^k))$ \\
	Step 5. 	\> Set $p^{\,\vec{e}}(x_1^k,\dots,x_\dimensionofmpp^k) = z_k - z_{k-1}$ \\
	Step 6. 	\> Go to Step 1
\end{tabbing}

\subsubsection*{Numerical Example}
We consider an example illustrating the computation of extreme measures with Poisson marginal distributions in the case  $J=3$. 
We explore their support, joint-probabilities, and resulting correlations. Henceforth, we shall refer to the extreme measures of a $3$-dimensional  Poisson process with intensities 
$\boldsymbol\mu = (\mu_1, \mu_2,\mu_3) = (3,5,7)$ as the ``extreme measure example''. 
We note that the tolerance level for the marginal distributions is $\varepsilon=0.01$.

We begin with the support of the distributions. As in the case $J=2$, the  support of an extreme measure looks like a staircase and is sparse. Figure \ref{chiu_fig:multi_d_support} illustrates the supports of all four extreme measures of the example,  where the associated monotonicity structures of the extreme measures are $\vec{e}^{\, 1}$, $\vec{e}^{\, 2}$, $\vec{e}^{\, 3}$, $\vec{e}^{\, 4}$:
\begin{enumerate}
	\item $\vec{e}^{\, 1} = (0,0,0)$ corresponds to the extreme measure in which all component exhibit extreme positive correlation
	\item $\vec{e}^{\, 2} = (0,1,0)$ corresponds to the extreme measure in which the second component has extreme negative correlation with the other coordinates
	\item $\vec{e}^{\, 3} = (0,0,1)$ corresponds to the extreme measure in which the third component has extreme negative correlation with the other coordinates
	\item $\vec{e}^{\, 4} = (0,1,1)$ corresponds to the extreme measure in which the first component has extreme negative correlation with the other coordinates
\end{enumerate}

Recall that the number of extreme measures for a given dimension $J$ is $N = 2^{J-1} = 4$ in this case (Lemma \ref{lem:numb_ed}). We also refer to extreme measures as extreme points. We display the N=4 extreme measures in blue in Figure 2. To highlight the monotonicity of the support of each extreme measure, we also show in Figure 2, its 2-dimensional projections onto the x-y, x-z and y-z planes.

\begin{figure}[H]
	\hspace{-1.5cm}\includegraphics[width=1.3\linewidth, height=10cm]{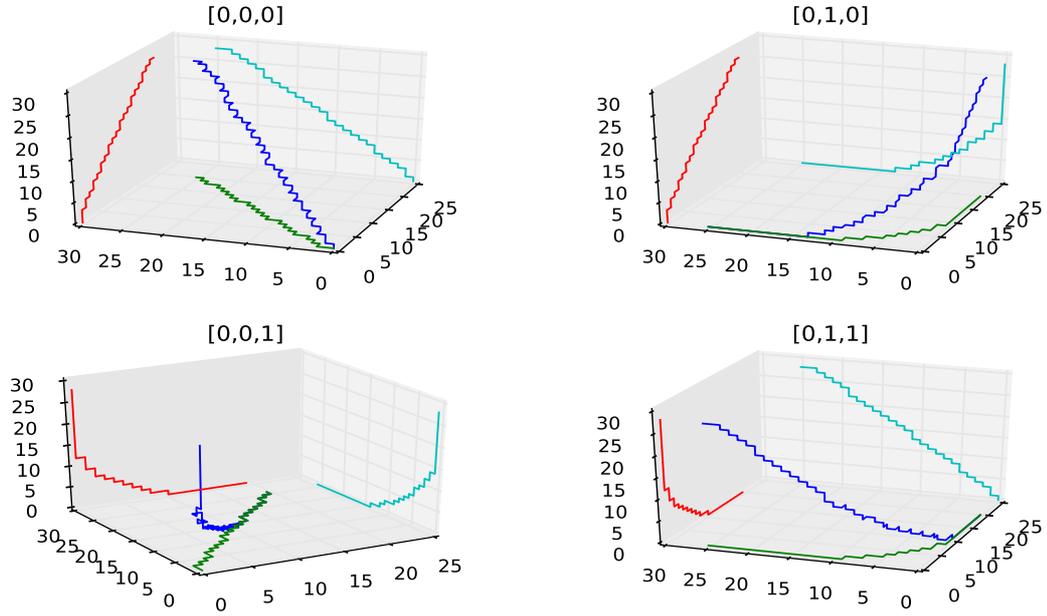}
	\captionsetup{width=0.85\textwidth}
	\caption{The blue curve in each graph is the support of an extreme measure in the case $J=3$, with Poisson marginal distributions. The red, teal, and green curves represent the projection of the 3D support onto the x-y, x-z, and y-z planes. These four graphs completely describe the support of the extreme measures in the case $J=3$.}
	\label{chiu_fig:multi_d_support}
\end{figure}

The resulting extreme correlation matrices are as follows:
\begin{equation*}
{\mathbf C^{\vec{e}^{\, 1}}}=
\begin{pmatrix}
1.0 & 0.93688 & 0.931861  \\
0.93688 & 1.0 & 0.967188 \\
0.931861 & 0.967188& 1.0 \\
\end{pmatrix}
\end{equation*}
\begin{equation*}
{\mathbf C^{\vec{e}^{\, 2}}}=
\begin{pmatrix}
1.0 & -0.81193 & 0.931861  \\
-0.81193 & 1.0 & -0.90135 \\
0.931861 & -0.90135 & 1.0 \\
\end{pmatrix}
\end{equation*}
\begin{equation*}
{\mathbf C^{\vec{e}^{\, 3}}}=
\begin{pmatrix}
1.0 & 0.93688 & -0.84624  \\
0.93688 & 1.0 & -0.90135 \\
-0.84624 & -0.90135 & 1.0 \\
\end{pmatrix}
\end{equation*}
\begin{equation*}
{\mathbf C^{\vec{e}^{\, 4}}}=
\begin{pmatrix}
1.0 & -0.81193 & -0.84624  \\
-0.81193 & 1.0 & 0.967188 \\
-0.84624 & 0.967188& 1.0 \\
\end{pmatrix}
\end{equation*}
where $C^{\vec{e}^{\, i}}$ is the  correlation matrix corresponding to the  monotonicity structure defined by the vector $\vec{e}^{\, i}, (i=1,2,3,4)$.

\begin{table}[H]
\centering
\begin{tabular}{|c|c|c|c|}
\hline
\multicolumn{2}{|c|}{Extreme Measure 1} & \multicolumn{2}{|c|}{Extreme Measure 2}  \\
\hline
$(i_1,i_2,i_3)$ & $p^{\vec{e}^{\, 1}}(i_1,i_2,i_3)$ & $(i_1,i_2,i_3)$ & $p^{\vec{e}^{\, 2}}(i_1,i_2,i_3)$\\
\hline
(0,0,0) & 0.0009 & (0,10,0) & 0.0000 \\
(0,0,1) & 0.0058 & (0,9,0) & 0.0002 \\
(0,1,1) & 0.0006 & (0,8,0) & 0.0009 \\
(0,1,2) & 0.0223 & (0,7,0) & 0.0034 \\
(0,1,3) & 0.0108 & (0,6,0) & 0.0120 \\
(0,2,3) & 0.0094 & (0,5,0) & 0.0332 \\
(1,2,3) & 0.0320 & (0,5,1) & 0.0029 \\
(1,2,4) & 0.0429 & (0,4,1) & 0.0902 \\
(1,3,4) & 0.0483 & (0,3,1) & 0.0563 \\
(1,3,5) & 0.0262 & (0,3,2) & 0.1242 \\
(2,3,5) & 0.0659 & (0,2,2) & 0.0446 \\
(2,4,5) & 0.0357 & (1,2,2) & 0.0553 \\
(2,4,6) & 0.1225 & (1,2,3) & 0.1708 \\
(3,4,6) & 0.0173 & (1,1,3) & 0.0532 \\
(3,5,6) & 0.0092 & (1,1,4) & 0.0885 \\
(3,5,7) & 0.1490 & (2,1,4) & 0.0795 \\
(3,5,8) & 0.0172 & (2,1,5) & 0.0494 \\
(3,6,8) & 0.0313 & (2,0,5) & 0.0514 \\
(4,6,8) & 0.0819 & (2,0,6) & 0.0036 \\
(4,6,9) & 0.0331 & (3,0,6) & 0.0468 \\
(4,7,9) & 0.0531 & (3,0,7) & 0.0145 \\
(5,7,9) & 0.0152 & (4,0,7) & 0.0071 \\
(5,7,10) & 0.0361 & (4,0,8) & 0.0081 \\
(5,8,10) & 0.0349 & (4,0,9) & 0.0001 \\
(5,8,11) & 0.0146 & (5,0,9) & 0.0026 \\
(6,8,11) & 0.0158 & (5,0,10) & 0.0005 \\
(6,9,11) & 0.0147 & (6,0,10) & 0.0003 \\
(6,9,12) & 0.0198 & (6,0,11) & 0.0002 \\
(7,9,12) & 0.0017 & (7,0,11) & 0.0000 \\
(7,10,12) & 0.0048 & (7,0,12) & 0.0001 \\
\hline
\end{tabular}
\captionsetup{width=0.85\textwidth}
\caption{Support and joint probabilities of the extreme measure corresponding to monotonicity structures $\vec{e}^{\, 1}$ and $\vec{e}^{\, 2}$}
\label{table:3d_ex_1}
\end{table}

\begin{table}[H]
\centering
\begin{tabular}{|c|c|c|c|}
\hline
\multicolumn{2}{|c|}{Extreme Measure 3} & \multicolumn{2}{|c|}{Extreme Measure 4}  \\
\hline
$(i_1,i_2,i_3)$ & $p^{\vec{e}^{\, 3}}(i_1,i_2,i_3)$ & $(i_1,i_2,i_3)$ & $p^{\vec{e}^{\, 4}}(i_1,i_2,i_3)$\\
\hline
(0,0,13) & 0.0000 & (0,10,13) & 0.0000 \\
(0,0,12) & 0.0001 & (0,10,12) & 0.0000 \\
(0,0,11) & 0.0002 & (0,9,12) & 0.0000 \\
(0,0,10) & 0.0008 & (0,9,11) & 0.0002 \\
(0,0,9) & 0.0027 & (0,8,11) & 0.0001 \\
(0,0,8) & 0.0081 & (0,8,10) & 0.0008 \\
(0,0,7) & 0.0216 & (0,7,10) & 0.0000 \\
(0,0,6) & 0.0504 & (0,7,9) & 0.0027 \\
(0,0,5) & 0.0514 & (0,7,8) & 0.0007 \\
(0,1,5) & 0.0494 & (0,6,8) & 0.0074 \\
(0,1,4) & 0.1680 & (0,6,7) & 0.0047 \\
(0,1,3) & 0.0151 & (0,5,7) & 0.0169 \\
(1,1,3) & 0.0381 & (0,5,6) & 0.0191 \\
(1,2,3) & 0.1708 & (0,4,6) & 0.0313 \\
(1,2,2) & 0.0999 & (0,4,5) & 0.0590 \\
(1,3,2) & 0.0591 & (0,3,5) & 0.0419 \\
(2,3,2) & 0.0651 & (0,3,4) & 0.1386 \\
(2,3,1) & 0.0563 & (0,2,4) & 0.0294 \\
(2,4,1) & 0.0626 & (0,2,3) & 0.0151 \\
(3,4,1) & 0.0276 & (1,2,3) & 0.2089 \\
(3,5,1) & 0.0029 & (1,2,2) & 0.0172 \\
(3,5,0) & 0.0308 & (1,1,2) & 0.1418 \\
(4,5,0) & 0.0024 & (2,1,2) & 0.0651 \\
(4,6,0) & 0.0120 & (2,1,1) & 0.0638 \\
(4,7,0) & 0.0009 & (2,0,1) & 0.0550 \\
(5,7,0) & 0.0026 & (3,0,1) & 0.0305 \\
(5,8,0) & 0.0005 & (3,0,0) & 0.0308 \\
(6,8,0) & 0.0004 & (4,0,0) & 0.0153 \\
(6,9,0) & 0.0002 & (5,0,0) & 0.0031 \\
(7,9,0) & 0.0000 & (6,0,0) & 0.0005 \\
\hline
\end{tabular}
\captionsetup{width=0.85\textwidth}
\caption{Extreme measures corresponding to monotonicity structures $\vec{e}^{\, 3}$ and $\vec{e}^{\, 4}$}
\label{table:3d_ex_2}
\end{table}

In Tables \ref{table:3d_ex_1} \& \ref{table:3d_ex_2}, we list the values of the joint probabilities $p^{\vec{e}}(\vec{i})$ for the extreme measures. Each table contains 2 of the 4  extreme measures. The columns are grouped such that they display the support and the corresponding joint probabilities corresponding to each example extreme measure. 

\section{Calibration of Correlations}
\label{sec:calibration}
In the case $J = 2$, given a correlation coefficient $\rho$ in the admissible correlation range $[\rho^*, \rho^{**}]$, we can use the following approach to find a probability measure $p$ having correlation $\rho$ and satisfying the marginal constraints. The approach is as follows. First find the unique $w \in [0,1]$ such that
\begin{equation*}
     \rho = w \rho^* + (1-w) \rho^{**}
\end{equation*}
Then set $p = w\, p^* + (1 - w)\, p^{**}$, where $p^*$ and $p^{**}$ are the extreme measures with correlations $\rho^*$ and $\rho^{**}$, respectively. Note that $p$ has correlation $\rho$ and that it also satisfies the marginal constraints, as it is a convex combination of $p^*$ and $p^{**}$, both of which also satisfy the marginal constraints. Note also that, if $\rho$ is not in the admissible correlation range $[\rho^*, \rho^{**}]$, then it cannot be the correlation of a probability measure $p$ satisfying the marginal constraints. 

If $J>2$, the same idea is applicable. However, we have instead, a system of equations with $N_w$ weights to solve for a given correlation matrix
\begin{equation}
\mathcal{C}_g	 = w_1\,\mathcal{C}_1 + \dots + w_{N_w}\,\mathcal{C}_{N_w},
	\label{eq:calibration_n_dim_corr}
\end{equation}
where the $\mathcal{C}_n$ are correlation matrices associated with the extreme distributions, $w_n \geq 0$ and $\sum_{n=1}^{N_w} w_n = 1$. Taking the extreme measures with the same set of marginal distributions, we construct the convex combination 
\begin{equation}
	\label{eq:calibration_n_dim_measure}
	p^w = w_1 p^{e_1} \, + \dots + w_{N_w} p^{e_{N_w}}
\end{equation}
where $p^w$ has correlation matrix $\mathcal{C}_g$ and satisfies the marginal constraints. The calibration problem is now reduced to finding a minimal $N_w$ to form a convex combination of extreme measures. Indeed, the number of extreme measures is $2^{J-1}$ and the number of correlation coefficients is $M = J(J-1)/2$. In matrix form (\ref{eq:calibration_n_dim_corr}) can be written as
\begin{equation}
	Aw = \hat{\mathcal{C}}_g \\
\end{equation}
where $A$ is of dimension $\dimensionofcorrelationspace$-by-$\numberextremepoints$, the $i^{\mathrm{th}}$ column of $A$ is a vectorized version of the upper triangular part of the extreme correlation matrix $\mathcal{C}_i$ and
$\hat{\mathcal{C}}_g$ is a vectorized version of the matrix $\mathcal{C}_g$. As the dimensionality of the multivariate Poisson process $J$ increases, $A$ becomes increasingly underdetermined. To find the weights, $w_j$, one can solve the following constrained system of equations
\begin{align}\label{eq:calibration_problem}
	& Aw = \hat{\mathcal{C}}_g  	\\
	& \textbf{1}^T w = 1 \nonumber \\
	& w_n \geq 0 \qquad  n=1, 2, \dots, N. \nonumber
\end{align}
An approach to solving (\ref{eq:calibration_problem}) is outlined on pages 376-379 of \cite{nocedal}. If (\ref{eq:calibration_problem}) does not have a solution, this implies that the correlation matrix $\mathcal{C}_g$ cannot be generated from a multivariate Poisson process with the prescribed marginal distributions. Once we have found a $w$ satisfying the constraints (\ref{eq:calibration_problem}), we can reduce the number of nonzero components in $w$ to $N_w \leq M+1$ using, for example a technique similar to that often used in the proof of Carath\'{e}odory's theorem, to obtain a vector of $N_w$ nonzero weights satisfying (\ref{eq:calibration_n_dim_corr}) and the positivity constraints on $w$.

A matrix $\mathcal{C}$ is called admissible if it is a symmetric, positive semi-definite (PSD) matrix with ones on the diagonal and each entry satisfies $\rho_{ij}^* \leq c_{ij} \leq \rho_{ij}^{**}$, where
$\rho_{ij}^*$ and $ \rho_{ij}^{**}$ are extreme correlations for the $2$-dimensional problem for $(X_i, X_j)$. 
Notice that the correlation matrices corresponding to the extreme measures are admissible.
\begin{theorem}
	\label{thm:convex_combo_valid_corr_matrix}
	A convex combination of admissible correlation matrices  is also an admissible correlation matrix.
\end{theorem}
\begin{proof}
	This fact readily follows from the observation that a convex combination of
    PSD matrices is a PSD matrix and, if all the matrices satisfy the correlation constraints so will the the convex combination of matrices.
\end{proof}
The probabilities $p^{\vec{e}}(\vec{i} )$ describing the extreme measures and their supports are  very different from  the  case of independent r.v.'s $X_j$. In particular, if $\rho=0$, the support of the  measure $p^w$ is the union of the supports of  $p^{\vec{e}}$. By adding an additional edge point $p^0$ corresponding to the case of independent components of $\vec{X}$, one can obtain a more general solution. We do not discuss this problem further in this paper.

\subsection*{Example Calibration}
We continue with the example extreme measure (c.f. Section \ref{sec:algor}) and attempt to calibrate to a target correlation matrix, $C_\ast$, given by
 \begin{equation*}
C_{\ast}= 
\begin{pmatrix}
1.0 & -0.8 & -0.5  \\
-0.8 & 1.0 & 0.5 \\
-0.5 & 0.5& 1.0 \\
\end{pmatrix}
\end{equation*}
Recall that in our example  $J = 3$, the Poisson marginal distributions have intensities $\boldsymbol\mu = (\mu_1,\mu_2,\mu_3) = (3,5,7)$ and $N = 2^{J-1} = 4$ extreme points. In this case, the constrained system corresponding to  can be constructed from the unique entries of the correlation matrix corresponding to each extreme point of the example extreme measure given in Section \ref{sec:algor} and takes the following form:
\begin{equation*}
\begin{pmatrix}
0.93688  & -0.81193 & 0.93688 & -0.81193 \\
0.931861 & 0.931861 & -0.84624 & -0.84624 \\
0.967188 & -0.90135 & -0.90135 & 0.967188 \\
1 & 1 & 1 & 1
\end{pmatrix}
\begin{pmatrix}
w_1 \\
w_2 \\
w_3 \\
w_4
\end{pmatrix}
=
\begin{pmatrix}
-0.8 \\
-0.5 \\
0.5 \\
1
\end{pmatrix}
\end{equation*}
A unique solution to this is $w$ = (0.0287993, 0.205588, 0.0436342, 0.721979). 
Now let
\begin{equation*}
p_* = w_1 \cdot p^{\vec{e}^1} 
        +  w_2 \cdot p^{\vec{e}^2}
        +  w_3 \cdot p^{\vec{e}^3}
        +  w_4 \cdot p^{\vec{e}^4}
\end{equation*}
where $p^{\vec{e}^i}$ is the extreme measure associate with the extreme correlation matrix $C^{\vec{e}^i}$,
$i = 1,2,3,4$, listed in Section \ref{sec:algor}.
Note that $p_*$ has correlation matrix $C_*$ and $p_*$ also satisfies the marginal constraints, since each of $p^{\vec{e}^i}$, $i = 1,2,3,4$, satisfies the marginal constraints.

\section{Simulation}
\label{sec:simulation}
Up until this point, we have discussed the computation of the multivariate Poisson distribution at some terminal time $T$ via the EJD algorithm. That allows us to achieve extreme correlations between the components of the multivariate Poisson process at  time $T$. We also obtain bounds on the elements of the admissible correlation matrix. The computation of the extreme measures allows us  to construct any admissible multivariate Poisson process.\footnote{That is a multivariate
Poisson process with correlations between the components satisfying the admissible correlation bounds.} 

We briefly discuss the BS approach, which allows us to simulate the  correlated multivariate Poisson processes on $[0,T]$ having an admissible correlation matrix at time $T$. Finally, we introduce the Forward continuation of the BS method. This extension allows us to construct sample paths of the multivariate Poisson process on the whole time axis.

\subsection*{Backward Simulation}
There are two general approaches to simulation of the sample paths of multivariate Poisson processes---a Forward approach and a Backward approach. Under the Forward simulation approach, the Frechet-Hoeffding theorem can be used to generate the inter-arrival times of the components. The correlation boundaries for the components of the multivariate Poisson process are tighter than the correlation boundaries attained using the BS approach. Furthermore, the time structure of correlation is richer in the Backward case. See \cite{kreinin} for a more detailed comparison.

The Backward approach relies on the conditional uniformity of the arrival times of the Poisson processes.
More precisely, the conditional distribution of the (unordered) arrival moments, $T_i$, of a Poisson process in the interval $[0,T]$, conditional on the 
number of events in the interval is uniform \cite{Feig}. 

The converse statement characterizing, the class of Poisson processes, is the foundation of the BS method  \cite{kreinin}. 
Consider a process $N_t, (t\ge 0)$ defined as
\begin{equation*}
N_t = \sum^{N_*}_{i=1} \mathbbm{1}(T_i \leq t), \quad 0 \leq t \leq T,
\end{equation*}
where $T_i$ are independent, identically distributed random variables uniformly distributed in the interval $[0, T]$.
Notice that $N_T = N_*$.

\begin{theorem}\label{thm_1d}
Let $N_*$ have a Poisson distribution with parameter $\lambda T$. Then $N_t$ is a Poisson process with intensity $\lambda$ in the interval $[0,T]$.
\end{theorem}
 
 Let us now formulate the generalization of Theorem~\ref{thm_1d}.
Suppose that coordinates of the random vector
$$  N_\ast =\bigl( N^{(1)} , \dots, N^{(J)} \bigr) $$ 
have Poisson distribution, 
$N^{(j)}_\ast \sim\Pois(\lambda_j T)$. 
Denote the correlation coefficient of $N^{(i)}_\ast$ and $N^{(j)}_\ast$ by $\rho_{ij}$. 

\begin{theorem}
Consider the processes
$$ N^{(j)}_t = \sum_{i=1}^{N^{(j)}_\ast} {\mathbbm{1}}{\bigl(T_i^{(j)}\le t\bigr)}, \,\,j=1, 2, \dots, J, 
$$
where the random variables, $T_i^{(j)}, (i=1, 2, \dots, N^{(j)}_\ast)$,  are mutually independent,
uniformly distributed in the interval $[0, T]$.
Then $\mathbf{N_t} =\bigl( N^{(1)}_t , \dots, N^{(J)}_t \bigr) $ is  a multivariate Poisson processes in the interval $[0, T]$ and
\begin{equation}
    \corr(N^{(i)}_t, N^{(j)}_t) = \rho_{ij} t T^{-1}, \quad 0\le t\le T.
\label{eq_corr_t}
\end{equation}
\end{theorem}
The proof can be found in \cite{kreinin}.
Let us now formulate the BS method:

\begin{itemize}
    \item[1.]Given a finite  vector of weights, $w_n$, $(n=1, 2, \dots, N_w)$, satisfying the conditions $w_n\ge 0$, $\sum_{n=1}^{N_w} w_n=1$, generate an index, $n$ by sampling from the probability distribution defined by $w$ to choose an extreme measure, $p^{\vec{e_n}}$.
    \item[2.] Generate a random vector $\mathbf{N_T}=(N_T(1), \dots N_T(J))$ from the extreme measure $p^{\, e_n}$.
    \item[3.] Generate arrival moments of the multivariate process $\mathbf{N_t}, (0\le t\le T)$. This can be accomplished via  
    straightforward simulation of the uniform distribution and ordering in the ascending order of the resulting samples  of the random variables $T_j$. 
\end{itemize}

\subsection*{Forward Continuation of the Backward Simulation}
The BS technique allows for the construction of sample paths of a multivariate Poisson process in an interval, $[0, T]$. In this section we consider an extension of the technique, which we call Forward-Backward simulation. We outline this approach for $J=2$.

Consider a sequence of time intervals $[0, T)$, $[T, 2T)$, $\dots, [mT, (m+1)T]$. Suppose that a bivariate Poisson process, $(X_t, Y_t)$, has already been simulated in the interval $[0, T)$ using the BS technique. For any $\tau$, $0\le\tau < T$, the increments $X_{T+\tau}- X_T$ are independent of $X_T$ and $Y_{T+\tau}-Y_T$ are independent of $Y_T$. Let us define the joint distribution of the increments as 
\begin{equation*}
	(X_{T+\tau} - X_T, Y_{T+\tau}- Y_T) \eqod (\hat X_\tau, \hat Y_\tau), \quad 0<\tau\le T,
\end{equation*}
where $\hat X_\tau$ and $\hat Y_\tau$ are independent versions of $X_t$ and $Y_t$, respectively, $\hat X_\tau \eqod X_t$ and $\hat Y_\tau \eqod Y_t$. Then we find $$\cov(X_{T+\tau}, Y_{T+\tau}) = \cov(X_{T}, Y_{T}) + \cov(X_{\tau}, Y_{\tau}).$$ Taking into account that
\begin{equation*}
	\cov(X_{\tau}, Y_{\tau})=\cov(X_T, Y_T) \cdot \frac{\tau^2}{T^2}, 
\end{equation*}
we obtain
\begin{equation*}
	\rho(T+\tau)=\rho(T) \frac{T^2+\tau^2 }{T(T+\tau) }. 
\end{equation*}
In particular, we have $\rho(2T)=\rho(T)$ and 
$\cov(X_{2T}, Y_{2T})=2\cov(X_{T}, Y_{T})$. 
Suppose now that $\rho(t)$ is defined for all $t\le nT$. 

Consider now the case $t=nT+\tau\in [nT, (n+1)T)$. We have
$$\cov(X_{nT}, Y_{nT})=n\cov(X_{T}, Y_{T})$$ 
and
\begin{equation*}
	\cov(X_{nT+\tau}, Y_{nT+\tau})=\cov(X_{T}, Y_{T})
   \cdot \Big(n +\frac{\tau^2}{T^2}\Big).
\end{equation*}
This latter relation implies
\begin{equation*}
	\rho(nT+\tau)=\rho(T) \frac{n+\tau^2\cdot T^{-2} }{n+\tau T^{-1} },
\end{equation*}
and we obtain asymptotic stationarity of the correlation coefficient:
$$ \lim_{n\to\infty} \rho(nT+\tau)=\rho(T) \quad \text{for all } \tau\in[0, T]. $$
Thus, the processes $X_t$ and $Y_t$ exhibit asymptotically stationary correlations as $t\rightarrow\infty$. 
An illustration of this is shown in Figure \ref{chiu_fig:fwd_cont}, where maximal (red line) and minimal
(blue line) values of
the correlation coefficient, $\corr(X_t, Y_t)$ are depicted.

It would be interesting to generalize this result for the class of mixed Poisson processes. The main difficulty is that the increments of the mixed Poisson processes are not independent.

\begin{figure}[H]
	\includegraphics[width=1\linewidth, height=9cm]{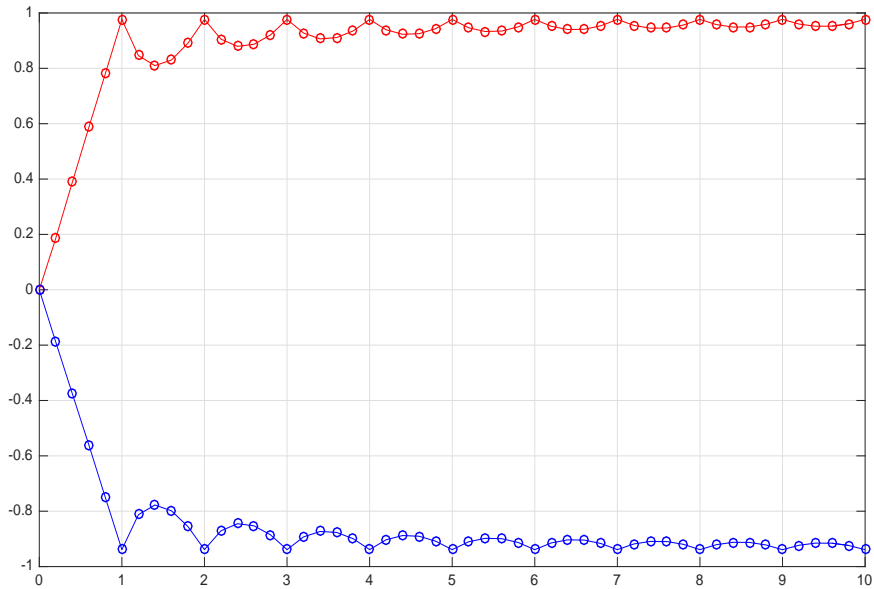}
	\caption{Forward Continuation of Backward Simulation: 
    $\corr(X_t, Y_t)$, $\mu_1 = 3$, $\mu_2 = 5$. }
	\label{chiu_fig:fwd_cont}
\end{figure}  	\section{Final remarks}
\label{sec:chiu_future_work}
We presented an approach to the solution to the problem of simulation of multivariate Poisson processes in the case the dimension of the 
problem is $J > 2$ and we described the admissible parameters for the calibration problem. We also
extended the BS approach  with the introduction of the Forward Continuation of BS.

There are several directions for future research. One is to extend the EJD approach to more general processes such as the Mixed Poisson processes and even to multivariate jump-diffusion processes. Another avenue of future research may be concerned with the efficient solutions of the multivariate calibration problem.  

It would also be interesting to study the interplay between the optimization problem and the EJD algorithm for 
computing the probabilities of the extreme measures and find the interpretation of this algorithm in terms of the Optimal transport problem. Exploring the synthesis of Forward and Backward simulation for more general processes is also worthwhile.

\end{document}